%% file: main.tex
\documentclass[letter,11pt]{article}

\usepackage{setspace}
\usepackage{subcaption}
\usepackage{enumerate}
\usepackage{amsmath}
\usepackage{amsfonts}
\usepackage{graphicx}
\usepackage{tcolorbox}
\usepackage{amssymb}
\usepackage{multirow}
\usepackage{geometry}
\usepackage{soul}
\usepackage{colortbl}
\usepackage{wrapfig}
\usepackage{algorithm}
\usepackage{algorithmicx}
\usepackage{algpseudocode}
\usepackage{amsthm}
\usepackage{todonotes}
\usepackage{mathtools}
\usepackage{mathrsfs}

\usepackage[pdftex, plainpages = false, pdfpagelabels, 
                 bookmarks=false,
                 bookmarksopen = true,
                 bookmarksnumbered = true,
                 breaklinks = true,
                 linktocpage,
                 pagebackref,
                 colorlinks = true,  
                 linkcolor = blue,
                 urlcolor  = blue,
                 citecolor = red,
                 anchorcolor = green,
                 hyperindex = true,
                 hyperfigures
                 ]{hyperref}

\usepackage{thmtools} 
\usepackage{thm-restate}
\usepackage{bm}

\algtext*{EndWhile}
\algtext*{EndIf}
\algtext*{EndFor}

\newtheorem{lemma}{Lemma}

\newtheorem{theorem}[lemma]{Theorem}
\newtheorem{theorem*}[lemma]{Theorem*}

\newtheorem{observation}[lemma]{Observation}

\newtheorem{fact}[lemma]{Fact}

\title{An $O(n \log n)$-Time Approximation Scheme for Geometric Many-to-Many Matching}
\author{Sayan Bandyapadhyay\footnote{Portland State University, USA. \texttt{sayanb@pdx.edu}.} \and Jie Xue\footnote{New York University Shanghai, China. \texttt{jiexue@nyu.edu}.}}
\date{}

\def\subparagraph{\paragraph}

\begin{document}

\maketitle

\bibliographystyle{plainurl}

\begin{abstract}
\input{abstract}
\end{abstract}

\section{Introduction}
\input{intro}

\section{Preliminaries} \label{sec-pre}
\input{pre}

\section{The approximation scheme} \label{sec-algo}
\input{algo}

\section{Conclusion and future work} \label{sec-conclusion}
\input{conclusion}

\bibliography{my_bib}

\end{document}

%% file: abstract.tex
Geometric matching is an important topic in computational geometry and has been extensively studied over decades.
In this paper, we study a geometric-matching problem, known as \textit{geometric many-to-many matching}.
In this problem, the input is a set $S$ of $n$ colored points in $\mathbb{R}^d$, which implicitly defines a graph $G = (S,E(S))$ where $E(S) = \{(p,q): p,q \in S \text{ have different colors}\}$, and the goal is to compute a minimum-cost subset $E^* \subseteq E(S)$ of edges that cover all points in $S$.
Here the cost of $E^*$ is the sum of the costs of all edges in $E^*$, where the cost of a single edge $e$ is the Euclidean distance (or more generally, the $L_p$-distance) between the two endpoints of $e$.
Our main result is a $(1+\varepsilon)$-approximation algorithm with an optimal running time $O_\varepsilon(n \log n)$ for geometric many-to-many matching in any fixed dimension, which works under any $L_p$-norm.
This is the first near-linear approximation scheme for the problem in any $d \geq 2$.
Prior to this work, only the bipartite case of geometric many-to-many matching was considered in $\mathbb{R}^1$ and $\mathbb{R}^2$, and the best known approximation scheme in $\mathbb{R}^2$ takes $O_\varepsilon(n^{1.5} \cdot \mathsf{poly}(\log n))$ time.

%% file: intro.tex
A central topic in computational geometry is the study of optimization problems on edge-weighted graphs that are defined geometrically (sometimes known as \textit{geometric graphs}).
Typically, geometric graphs use points in $\mathbb{R}^d$ as their vertices, and the Euclidean distance (or distance under other norms) between two points naturally defines the weight of an edge.
Many fundamental graph optimization problems have been investigated on geometric graphs, including minimum spanning tree~\cite{agarwal1991euclidean,march2010fast,shamos1975closest}, Steiner trees~\cite{bandyapadhyay2024euclidean,brazil2014history,smith1992find}, traveling salesman problem~\cite{arora1996polynomial,arora1997nearly,de2018eth}, spanners~\cite{arya1995euclidean,har2013euclidean,le2022truly}, matching~\cite{agarwal2022deterministic,varadarajan1998divide,varadarajan1999approximation}, etc.
By exploiting the underlying geometry, these problems can usually be solved much more efficiently on geometric graphs, compared to general edge-weighted graphs.

In this paper, we focus on a particularly important class of problems on geometric graphs, the matching-related problems.
These problems have applications in a wide range of areas, e.g., computational biology \cite{ben2002restriction}, data mining \cite{belhadi2020data}, computational music \cite{toussaint2004geometry,toussaint2004comparison}, machine learning \cite{ristoski2018machine}, etc.
When studying matching problems on geometric graphs, there are two settings commonly used in the literature.
The first one, called the \textit{bipartite} setting, requires the input points to be bichromatic and the graph considered is the complete bipartite graph consisting of edges between points with different colors.
The second one, called the \textit{complete} setting, simply considers the complete graph induced by the input points.

One of the matching problems that has received most attention on geometric graphs is the classical \textit{minimum-weight perfect matching} problem, where the goal is to compute a set of vertex-disjoint edges with minimum total weight that cover all vertices.
The problem can be solved in $O(nm+n^2 \log n)$ time on any graph with $n$ vertices and $m$ edges, by the seminal work of Fredman and Tarjan~\cite{fredman1987fibonacci}.
A popular line of research~\cite{agarwal2022deterministic,indyk2007near,kaplan2020dynamic,RaghvendraA20,vaidya1989geometry,varadarajan1998divide,varadarajan1999approximation} in computational geometry investigated minimum-weight perfect matching on geometric graphs, leading to much faster exact and approximation algorithms.
The exact algorithms were designed for the problem in $\mathbb{R}^2$.
The best algorithm for the bipartite setting~\cite{kaplan2020dynamic} runs in $O(n^2 \cdot \mathsf{poly}(\log n))$ time where $n$ is the number of input points, while the best algorithm for the complete setting~\cite{varadarajan1998divide} runs in $O(n^{1.5} \cdot \mathsf{poly}(\log n))$ time.
The approximation algorithms are much more general and efficient.
It was known that in both bipartite and complete settings, the problem admits $(1+\varepsilon)$-approximation algorithms with running time $O_\varepsilon(n \cdot \mathsf{poly}(\log n))$\footnote{Here $O_\varepsilon(\cdot)$ hides factors depending only on $\varepsilon$.} in $\mathbb{R}^d$ for any fixed $d$~\cite{agarwal2022deterministic,RaghvendraA20,varadarajan1999approximation}.

The minimum-weight perfect matching problem has an interesting variant, which is also a classical problem known as \textit{many-to-many matching} or \textit{edge cover} \cite{FerdousPK18,gallagher20194,keijsper1998efficient,ni2008fuzzy,norman1959algorithm,white1971minimum}.
In this variant, the only difference is that the edges in the solution need not to be vertex-disjoint.
In other words, it simply asks for a set of edges with minimum total weight that cover all vertices.
Many-to-many matching can be reduced to minimum-weight perfect matching~\cite{keijsper1998efficient}, and is thus polynomial-time solvable.
On geometric graphs, many-to-many matching, while having received less attention than minimum-weight perfect matching, also has a long history.
Eiter and Mannila~\cite{eiter1997distance} introduced the problem for the first time in 1997, under the name of \textit{link distance}, in order to mesure the similarity between two sets of points.
Colannino and Toussaint~\cite{colannino2005faster} considered geometric many-to-many matching in the bipartite setting and showed that the problem can be solved in $O(n^2)$ time in $\mathbb{R}^1$.
Later, Colannino et al.~\cite{ColanninoDHLMRST07} improved this result and obtained an optimal $O(n\log n)$-time algorithm, which completely settles the complexity of (bipartite) geometric many-to-many matching in $\mathbb{R}^1$.
Several variants of the problem in $\mathbb{R}^1$ have also been considered~\cite{Rajabi-AlniB16, abs-1904-05184, abs-1904-03015}, in which the input points can have capacities and/or demands.
Recently, Bandypadhyay et al.~\cite{bandyapadhyay2021exact} studied bipartite geometric many-to-many matching in $\mathbb{R}^2$ and designed two algorithms.
The first algorithm solves the problem in $O(n^2\cdot \mathsf{poly}(\log n))$ time.
This algorithm is based on the general reduction from many-to-many matching to perfect matching, and exploits various geometric data structures to implement the reduction and the Hungarian algorithm~\cite{kuhn1956variants} for perfect matching in an efficient way.
The second algorithm is a $(1+\varepsilon)$-approximation algorithm which runs in $O_{\varepsilon}(n^{1.5}\cdot \mathsf{poly}(\log n))$ time.
The basic idea of this algorithm is similar to the first one, but it uses the multi-scale algorithm of Gabow and Tarjan~\cite{gabow1989faster} for perfect matching instead of the Hungarian algorithm, which can be implemented more efficiently in the geometric setting by losing a factor of at most $1+\varepsilon$ in cost.
In higher dimensions, no nontrivial results for geometric many-to-many matching were known, to the best of our knowledge.

As one can see in the above discussion, in terms of exact algorithms, the best known bounds for geometric many-to-many matching are similar to the best known bounds for geometric (minimum-weight) perfect matching, at least in the bipartite setting --- both problems can be solved in near-quadratic time in $\mathbb{R}^2$~\cite{bandyapadhyay2021exact,kaplan2020dynamic} and are open in higher dimensions.
However, in terms of approximation algorithms, geometric many-to-many matching is much less well-understood than geometric perfect matching.
Even in $\mathbb{R}^2$, no approximation scheme for geometric many-to-many matching with near-linear running time was known, while geometric perfect matching admits near-linear approximation schemes in any fixed dimension~\cite{agarwal2022deterministic,RaghvendraA20,varadarajan1999approximation}.
This motivates the following natural question, which is the subject of this paper.

\begin{tcolorbox}[colback=gray!5!white,colframe=gray!75!black]
\textbf{Question:}
Does geometric many-to-many matching admit a $(1+\varepsilon)$-approximation algorithm with running time $O_\varepsilon(n \cdot \mathsf{poly}(\log n))$ in $\mathbb{R}^d$ for any fixed $d$?
\end{tcolorbox}

We answer this question affirmatively by giving such an approximation scheme for geometric many-to-many matching.
In fact, our results are much stronger and more general.
Our algorithm has an optimal $O_\varepsilon(n \log n)$ running time, works under any $L_p$-norm, and applies to both bipartite and complete settings (and beyond).
We shall discuss our results in detail in the next section.

\subsection{Our results}

We study geometric many-to-many matching in a \textit{colored} setting which simultaneously generalizes the aforementioned bipartite and complete settings for geometric graphs.
Here, each input point is associated with a color (the total number of colors can be unbounded) and the graph considered has edges between every pair of points with different colors.
Clearly, the colored setting is equivalent to the bipartite setting when there are only two colors and is equivalent to the complete setting when all points have distinct colors.
Let $S$ be a set of colored points in $\mathbb{R}^d$.
We write $E(S) = \{(p,q): p,q \in S \text{ have different colors}\}$ as the edge set of the geometric graph induced by $S$.
The \textit{$L_p$-cost} of an edge $e \in E(S)$ is the $L_p$-distance between its two endpoints, and the \textit{$L_p$-cost} of a subset $E \subseteq E(S)$ is the sum of the $L_p$-costs of all edges in $E$.
We say $E \subseteq E(S)$ \textit{covers} a point in $S$ if the point is an endpoint of an edge in $E$.
The \textit{geometric many-to-many matching} problem is formally defined as follows.

\begin{tcolorbox}[colback=gray!5!white,colframe=gray!75!black]
{\sc Geometric Many-to-Many Matching} \\
\textbf{Input:} A set $S$ of $n$ colored points in $\mathbb{R}^d$. \\
\textbf{Goal:} Compute $E^* \subseteq E(S)$ with minimum $L_p$-cost which covers all points in $S$.
\end{tcolorbox}

When studying the problem, we assume that the input points in $S$ are already clustered by their colors so that we do not need extra time to compute the color-partition of $S$.
Alternatively, one can assume the colors belong to $[n] = \{1,\dots,n\}$ and thus the color-partition of $S$ can be computed in linear time.
Our main result is the following theorem.

\begin{theorem} \label{thm-main}
    For any fixed $d \in \mathbb{N}$ and $p \geq 1$, geometric many-to-many matching in $\mathbb{R}^d$ under the $L_p$-norm admits a $(1+\varepsilon)$-approximation algorithm with running time $O_\varepsilon(n \log n)$.
\end{theorem}

Note that the running time in Theorem~\ref{thm-main} is \textit{optimal} up to a factor depending on $\varepsilon$.
Indeed, as observed in~\cite{ColanninoDHLMRST07}, any approximation algorithm for geometric many-to-many matching in $\mathbb{R}^1$ requires $\Omega(n \log n)$ time, due to a reduction from set equality.

Interestingly, our algorithm in Theorem~\ref{thm-main} completely bypasses the reduction from many-to-many matching to minimum-weight perfect matching.
This allows us to avoid the techniques for perfect matching such as augmenting paths, which were commonly used in the previous algorithms for geometric matching problems~\cite{agarwal2022deterministic,bandyapadhyay2021exact,RaghvendraA20}.
Instead, our algorithm exploits the nice structures of the many-to-many matching problem itself, and solves the problem by nontrivially combining Baker's shifting technique~\cite{baker1994approximation}, grid techniques, approximate nearest-neighbor search~\cite{chan2020locality}, and the FPT algorithm for integer linear programming~\cite{CyganFKLMPPS15}.

\subparagraph{Organization.}
In Section~\ref{sec-pre}, we introduce the basic notions needed throughout the paper.
Section~\ref{sec-algo} presents our main algorithm and proves Theorem~\ref{thm-main}.
Finally, we conclude the paper and pose some open questions in Section~\ref{sec-conclusion}.

%% file: pre.tex
\subparagraph{Basic notations.}
We use $\mathbb{N}$ to denote the set of natural numbers including 0.
For a number $n \in \mathbb{N}$, we write $[n] = \{1,\dots,n\}$.
A \textit{colored point} in $\mathbb{R}^d$ is a point $p \in \mathbb{R}^d$ with a color, which we denote by $\mathsf{cl}(p)$.
Let $S$ be a set of colored points in $\mathbb{R}^d$.
We define $E(S) = \{(p,q): p,q \in S \text{ and } \mathsf{cl}(p) \neq \mathsf{cl}(q)\}$ as the \textit{edge set} on $S$.
Here, the pairs in $E(S)$ are \textit{unordered}, i.e., $(p,q)$ and $(q,p)$ are viewed as one element in $E(S)$.
For a subset $E \subseteq E(S)$, we denote by $V(E) \subseteq S$ the subset consisting of the endpoints of the edges in $E$.

\subparagraph{Foreign neighbors.}
Let $S$ be a set of colored points in $\mathbb{R}^d$.
For a point $p \in S$, a \textit{foreign neighbor} of $p$ in $S$ refers to another point $q \in S$ satisfying $\mathsf{cl}(q) \neq \mathsf{cl}(p)$.
We say $q$ is a \textit{$c$-approximate nearest} foreign neighbor of $p$ in $S$ (with respect to a metric $\mathsf{dist}:\mathbb{R}^d \times \mathbb{R}^d \rightarrow \mathbb{R}_{\geq 0}$) if $\textsf{dist}(p,q) \leq c \cdot \textsf{dist}(p,q')$ for all foreign neighbors $q'$ of $p$ in $S$.
The following lemma is a direct consequence of the dynamic approximate nearest neighbor data structure of Chan et al.~\cite{chan2020locality}.
\begin{lemma} \label{lem-anfn}
Given a set $S$ of $n$ colored points in $\mathbb{R}^d$, one can compute in $O_\varepsilon(n \log n)$ time a function $\mathsf{ann}:S \rightarrow S$ which maps each point in $S$ to a $(1+\varepsilon)$-approximate nearest foreign neighbor (with respect to the Euclidean distance) of that point in $S$.
The algorithm generalizes to the $L_p$-norm for any fixed $p \geq 1$.
\end{lemma}
\begin{proof}
We shall apply the dynamic approximate nearest neighbor data structure $\mathcal{A}$ of Chan et al.~\cite{chan2020locality}.
The data structure $\mathcal{A}$ stores a dynamic set of points in $\mathbb{R}^d$, and supports $(1+\varepsilon)$-approximate nearest neighbor query to the point-set with query time $O_\varepsilon(\log n)$, where $n$ is the size of the current set.
Insertions and deletions are supported with update time $O_\varepsilon(\log n)$.
Suppose the points in $S$ have $r$ colors in total, and let $S_i \subseteq S$ consist of the points with the $i$-th color for $i \in [r]$.
We build the data structure $\mathcal{A}$ on $S$, which can be done in $O_\varepsilon(n \log n)$ time by inserting the points in $S$ to $\mathcal{A}$ one by one.
Then we consider each $i \in [r]$, and compute $\mathsf{ann}(p)$ for $p \in S_i$ as follows.
We first delete the points in $S_i$ from $\mathcal{A}$.
Then for every $p \in S_i$, we query $\mathcal{A}$ to obtain a $(1+\varepsilon)$-approximate nearest neighbor $q$ of $p$ in $S \backslash S_i$ and set $\mathsf{ann}(p) = q$, which is a $(1+\varepsilon)$-approximate nearest foreign neighbor of $p$ in $S$.
After this, we insert the points in $S_i$ back to $\mathcal{A}$.
In this way, we can obtain $\mathsf{ann}(p)$ for all $p \in S$.
The time cost for each $i \in [r]$ is $O_\varepsilon(|S_i| \log n)$, since both the query time and the update time of $\mathcal{A}$ are $O_\varepsilon(\log n)$.
The overall running time is then $O_\varepsilon(n \log n)$, because $\sum_{i=1}^r |S_i| = n$.
The data structure of Chan et al.~\cite{chan2020locality} works under any $L_p$-norm for $p \geq 1$, so does our algorithm.
\end{proof}

\subparagraph{Grids.}
A \textit{$d$-dimensional grid} refers to a (infinite) set of axis-parallel hyperplanes that partition $\mathbb{R}^d$ into same-sized axis-parallel hypercubes (called \textit{grid cells} or simply \textit{cells}).
A $d$-dimensional grid can be characterized by a number $w > 0$ called the \textit{cell-size} and a vector $(k_1,\dots,k_d) \in \mathbb{R}^d$ called the \textit{offset}.
Specifically, the grid with cell-size $w$ and offset $(k_1,\dots,k_d)$, denoted by $\varGamma_w(k_1,\dots,k_d)$, consists of the hyperplanes whose equations are of the form $x_i = wt + k_i$ for $i \in [d]$ and $t \in \mathbb{Z}$.
In other words, $\varGamma_w(k_1,\dots,k_d)$ is the grid in which the cells are hypercubes of side-length $w$ and $(k_1,\dots,k_d)$ is a grid point (i.e., a vertex of a cell).
Figure~\ref{fig-grid} presents the 2-dimensional grid $\varGamma_3(1,2)$.

\begin{figure}[h]
    \centering
    \includegraphics[height=5cm]{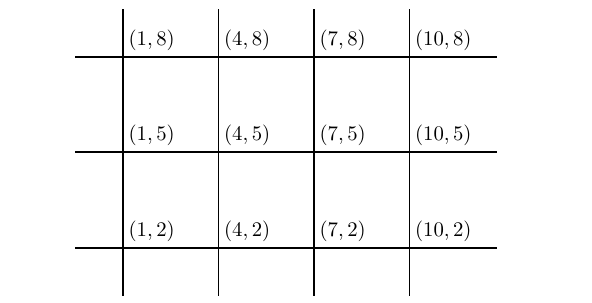}
    \caption{The 2-dimensional grid $\varGamma_3(1,2)$. The numbers are the coordinates of the grid points.}
    \label{fig-grid}
\end{figure}

For a set $S$ of points in $\mathbb{R}^d$, a ($d$-dimensional) grid naturally induces a partition of $S$ in which each part consists of the points in one grid cell.
To guarantee that every point in $\mathbb{R}^d$ belongs to exactly one grid cell, we define the cells as hypercubes that are closed on the lower side and open on the higher side.
Formally, every grid cell of $\varGamma_w(k_1,\dots,k_d)$ is a hypercube $\prod_{i=1}^d [t_iw+k_i, (t_i+1)w+k_i)$ where $(t_1,\dots,t_d) \in \mathbb{Z}^d$.

%% file: algo.tex
We present our approximation scheme under the Euclidean norm.
Its extension to any $L_p$-norm is straightforward, and will be briefly discussed in Section~\ref{sec-alltogether}.

Let $S$ be a set of $n$ colored points in $\mathbb{R}^d$.
For each point $p$, let $\mathsf{nn}(p) \in S$ be the nearest foreign neighbor of $p$ in $S$ (with respect to the Euclidean distance).
For simplicity of exposition, we shall first present our algorithm under the assumption that $\mathsf{nn}(p)$ for every $p \in S$ is known to us.
It is unlikely to compute all nearest foreign neighbors in near-linear time for $d \geq 3$, due to the conjectured $\Omega(n^{4/3})$ lower bound by Erickson~\cite{erickson1995relative}.
However, as we only want an approximation algorithm, it turns out that knowing \textit{approximate} nearest foreign neighbors (which can be computed efficiently using Lemma~\ref{lem-anfn}) is already sufficient.
We shall discuss this in Section~\ref{sec-alltogether}.

We define $\phi(p)$ as the distance between $p$ and $\mathsf{nn}(p)$ for all $p \in S$.
The feasible solutions for geometric many-to-many matching on $S$ are subsets $E \subseteq E(S)$ satisfying $V(E) = S$.
It is more convenient to consider an equivalent formulation of the problem where all subsets of $E(S)$ are feasible solutions, which we call the \textit{penalized} formulation.
In this formulation, we allow the solution to not cover all points in $S$, but for every uncovered point $p \in S$, there is a penalty of $\phi(p)$ added to the cost of the solution.
Formally, the cost of a solution $E \subseteq E(S)$ is $\sum_{e \in E} |e| + \sum_{p \in S \backslash V(E)} \phi(p)$, where $|e|$ denotes the \textit{length} of $e$, i.e., the Euclidean distance between the endpoints of $e$.
To see this formulation is equivalent to the original one, let $\mathsf{opt}$ (resp., $\mathsf{opt}'$) be the optimum of the original (resp., penalized) formulation of the problem.
Clearly, a solution of the original formulation is also a solution of the penalized formulation with the same cost, which implies $\mathsf{opt} \geq \mathsf{opt}'$.
The following lemma shows that $\mathsf{opt} \leq \mathsf{opt}'$.

\begin{lemma} \label{lem-penalized}
Given a subset $E \subseteq E(S)$, one can compute in $O(n+|E|)$ time another subset $E' \subseteq E(S)$ such that $V(E') = S$ and $\sum_{e \in E'} |e| \leq \sum_{e \in E} |e| + \sum_{p \in S \backslash V(E)} \phi(p)$.
\end{lemma}
\begin{proof}
We simply set $E' = E \cup \{(p,\mathsf{nn}(p)): p \in S \backslash V(E)\}$, which can be computed in $O(n+|E|)$ time (given the nearest foreign neighbors).
Clearly, $V(E') = S$.
For an edge $e = (p,\mathsf{nn}(p))$, we have $|e| = \phi(p)$.
Thus, $\sum_{e \in E'} |e| \leq \sum_{e \in E} |e| + \sum_{p \in S \backslash V(E)} \phi(p)$.
\end{proof}

As $\mathsf{opt} = \mathsf{opt}'$, a $c$-approximation solution for the original formulation is also a $c$-approximation solution for the penalized formulation.
On the other hand, Lemma~\ref{lem-penalized} shows that given a $c$-approximation solution for the penalized formulation, one can compute in linear time a $c$-approximation solution for the original formulation.
Thus, the two formulations are equivalent, and it suffices to solve the penalized formulation.

We sort the points in $S$ by their coordinates in every dimension, and also by their $\phi$-values.
The benefit of considering the penalized formulation is that it allows us to properly define \textit{subproblems}.
Specifically, for $R \subseteq S$ and $E \subseteq E(R)$, we write 
\begin{equation*}
    \mathsf{cost}_R(E) = \sum_{e \in E} |e| + \sum_{p \in R \backslash V(E)} \phi(p).
\end{equation*}
We define a subproblem $\mathbf{Prob}(R)$ for every $R \subseteq S$, which aims to compute $E \subseteq E(R)$ that minimizes $\mathsf{cost}_R(E)$.
Note that $\mathbf{Prob}(R)$ is not exactly equivalent to the (penalized) geometric many-to-many matching problem on $R$, as the penalty $\phi(p)$ for $p \in R$ is defined by the nearest foreign neighbor of $p$ in $S$ (rather than $R$).
We observe the following simple fact.


\begin{fact} \label{fact-edgelen}
Let $R \subseteq S$ and $E^* \subseteq E(R)$ be an optimal solution of $\mathbf{Prob}(R)$.
Then any edge $e = (p,q) \in E^*$ satisfies $|e| \leq \phi(p)+\phi(q)$.
\end{fact}
\begin{proof}
Assume there exists $e = (p,q) \in E^*$ with $|e| > \phi(p)+\phi(q)$.
Then 
\begin{equation*}
\mathsf{cost}_R(E^* \backslash \{e\}) \leq \mathsf{cost}_R(E^*) - |e| + \phi(p)+\phi(q) < \mathsf{cost}_R(E^*),    
\end{equation*}
contradicting the optimality of $E^*$.
\end{proof}

Clearly, our final goal is to compute a $(1+\varepsilon)$-approximation solution for $\mathbf{Prob}(S)$.
Without loss of generality, assume $\varepsilon \leq 1$.
For $R \subseteq S$, we denote by $\mathsf{opt}(R) = \min_{E \subseteq E(R)} \mathsf{cost}_R(E)$ as the optimum of subproblem $\mathbf{Prob}(R)$.
Our algorithm first applies two reductions, which eventually reduce $\mathbf{Prob}(S)$ to certain well-structured subproblems.
These reductions are presented in Sections~\ref{sec-first} and~\ref{sec-second}.
Then we use grid technique together with the FPT algorithm for integer linear programming to solve these subproblems, which is discussed in Section~\ref{sec-subprob}.
Finally, we put everything together and prove Theorem~\ref{thm-main} in Section~\ref{sec-alltogether}.

\subsection{First reduction} \label{sec-first}

In the first step, we reduce the problem $\mathbf{Prob}(S)$ to subproblems $\mathbf{Prob}(R)$ where the points in $R$ have similar $\phi$-values.
For $R \subseteq S$, let $\Delta_R = \max_{p \in R} \phi(p) / \min_{p \in R} \phi(p)$.
Our goal in this section is to prove the following lemma.

\begin{lemma} \label{lem-first}
Suppose for every subset $R \subseteq S$ satisfying $\Delta_R \leq (\frac{3}{\varepsilon})^{\lceil \frac{3}{\varepsilon} \rceil}$, one can compute in $O_\varepsilon(|R|)$ time a $(1+\frac{\varepsilon}{2})$-approximation solution for $\mathbf{Prob}(R)$.
Then one can compute in $O_\varepsilon(n)$ time a $(1+\varepsilon)$-approximation solution for $\mathbf{Prob}(S)$.
\end{lemma}

The basic idea of this reduction is the following.
We distinguish the edges in $E(S)$ as balanced edges and unbalanced edges.
The balanced edges are those whose two endpoints have similar $\phi$-values.
It turns out that the unbalanced edges can be ``ignored'' almost for free.
Regarding only the balanced edges, we can then apply Baker's shifting technique~\cite{baker1994approximation} on the $\phi$-values to decompose the problem.
Below we discuss the reduction in detail.

For each point $p \in S$, we write $\phi'(p) = \log_{3/\varepsilon} \phi(p)$.
We say an edge $e = (p,q) \in E(S)$ is \textit{balanced} if $|\phi'(p) - \phi'(q)| \leq 1$, and \textit{unbalanced} otherwise.

\begin{observation} \label{obs-unba}
For any unbalanced edge $e = (p,q) \in E(S)$, $\phi(p)+\phi(q) \leq (1+\frac{\varepsilon}{3}) \cdot |e|$.
\end{observation}
\begin{proof}
Without loss of generality, assume $\phi(p) \geq \phi(q)$.
As $e$ is unbalanced, $|\phi'(p) - \phi'(q)| > 1$ and thus $\phi(p) > \frac{3}{\varepsilon} \cdot \phi(q)$.
Note that $|e| \geq \phi(p)$ by the definition of $\phi(p)$.
Therefore, we have $\phi(p)+\phi(q) \leq (1+\frac{\varepsilon}{3}) \cdot \phi(p) \leq (1+\frac{\varepsilon}{3}) \cdot |e|$.
\end{proof}

Set $w = \lceil \frac{3}{\varepsilon} \rceil$.
For each $i \in [w]$, we construct the (1-dimensional) grid $\varGamma_w(i)$ --- recall that $\varGamma_w(i)$ is the grid with cell-size $w$ and offset $i$.
For an edge $e = (p,q) \in E(S)$, we define $I_e = \{i \in [w]: \phi'(p) \text{ and } \phi'(q) \text{ lie in the same cell of } \varGamma_w(i)\}$.

\begin{observation} \label{obs-largeIe}
For any balanced edge $e \in E(S)$, $|I_e| \geq w-1$.
\end{observation}
\begin{proof}
Suppose $e = (p,q)$.
Since $e$ is balanced, $|\phi'(p) - \phi'(q)| \leq 1$ and thus there exists at most one integer $i^* \in [\phi'(p),\phi'(q))$.
For any $i \in [w]$, $i \notin I_e$ iff $i^*$ exists and $i$ is congruent with $i^*$ modulo $w$.
Thus, $|I_e| = w$ if $i^*$ does not exist and $|I_e| = w-1$ if $i^*$ exists.
\end{proof}

For each $i \in [w]$, the grid $\varGamma_w(i)$ induces a partition of the values in $\{\phi'(p): p \in S\}$, which in turn induces a partition 
$\mathcal{R}_i$ of $S$.
In other words, $\mathcal{R}_i$ partitions $S$ into subsets each of which contains the points in $S$ whose $\phi'$-values lying in one cell of $\varGamma_w(i)$.

\begin{observation} \label{obs-onegoodidx}
    There exists $i \in [w]$ such that $\sum_{R \in \mathcal{R}_i} \mathsf{opt}(R) \leq (1+\frac{\varepsilon}{3}) \cdot \mathsf{opt}(S)$.
\end{observation}
\begin{proof}
Let $E^* \subseteq E(S)$ be an optimal solution of $\mathbf{Prob}(S)$.
Also, let $B^* \subseteq E^*$ and $U^* \subseteq E^*$ be the subsets consisting of balanced and unbalanced edges in $E^*$, respectively.
For $i \in [w]$, define $N_i^* = \{e \in B^*: i \notin I_e\}$.
By Observation~\ref{obs-largeIe}, each balanced edge $e \in B^*$ is contained in at most one set $N_i^*$.
Thus, we have $\sum_{i=1}^w \sum_{e \in N_i^*} |e| \leq \sum_{e \in B^*} |e|$, which implies the existence of an index $i \in [w]$ such that $\sum_{e \in N_i^*} |e| \leq (\sum_{e \in B^*} |e|)/w \leq \frac{\varepsilon}{3} \cdot (\sum_{e \in B^*} |e|)$.
We show that $\sum_{R \in \mathcal{R}_i} \mathsf{opt}(R) \leq (1+\frac{\varepsilon}{3}) \cdot \mathsf{opt}(S)$.
Note that $N_i^*$ contains exactly the balanced edges in $E^*$ whose two endpoints belong to different sets in $\mathcal{R}_i$.
Therefore, we have
\begin{equation*}
    \sum_{R \in \mathcal{R}_i}\mathsf{opt}(R) \leq \sum_{R \in \mathcal{R}_i} \mathsf{cost}_R(B^* \cap E(R)) \leq \mathsf{cost}_S(E^*) - \sum_{e \in U^* \cup N_i^*} |e| + \sum_{p \in V(U^* \cup N_i^*)} \phi(p).
\end{equation*}
By Observation~\ref{obs-unba}, $\sum_{p \in V(U^*)} \phi(p) -\sum_{e \in U^*} |e| \leq \frac{\varepsilon}{3} \sum_{e \in U^*} |e|$.
Furthermore, every $e = (p,q) \in E(S)$ satisfies $|e| \geq \max\{\phi(p),\phi(q)\} \geq \frac{1}{2} (\phi(p)+\phi(q))$ by the definition of $\phi$.
Therefore, $\sum_{e \in N_i^*} |e| \geq \frac{1}{2} \sum_{p \in V(N_i^*)} \phi(p)$.
It follows that $\sum_{p \in V(N_i^*)} \phi(p) -\sum_{e \in N_i^*} |e| \leq \sum_{e \in N_i^*} |e| \leq \frac{\varepsilon}{3} \sum_{e \in B^*} |e|$, where the second inequality follows from the choice of $i$.
Now,
\begin{align*}
    \sum_{R \in \mathcal{R}_i}\mathsf{opt}(R)\ \leq\ \ & \mathsf{cost}_S(E^*) - \sum_{e \in U^* \cup N_i^*} |e| + \sum_{p \in V(U^* \cup N_i^*)} \phi(p) \\
    \ \leq\ \  & \mathsf{cost}_S(E^*) + \left(\sum_{p \in V(U^*)} \phi(p) -\sum_{e \in U^*} |e|\right) + \left(\sum_{p \in V(N_i^*)} \phi(p) -\sum_{e \in N_i^*} |e|\right) \\
    \ \leq\ \  & \mathsf{cost}_S(E^*) + \frac{\varepsilon}{3} \sum_{e \in U^*} |e| + \frac{\varepsilon}{3} \sum_{e \in B^*} |e| \\
    \ \leq\ \ & \left(1+\frac{\varepsilon}{3}\right) \cdot \mathsf{cost}_S(E^*),
\end{align*}
where the last inequality follows from the fact $\sum_{e \in U^*} |e| + \sum_{e \in B^*} |e| = \sum_{e \in E^*} |e| \leq \mathsf{cost}_S(E^*)$.
As $\mathsf{cost}_S(E^*) = \mathsf{opt}(S)$, we conclude that $\sum_{R \in \mathcal{R}_i} \mathsf{opt}(R) \leq (1+\frac{\varepsilon}{3}) \cdot \mathsf{opt}(S)$.
\end{proof}

Using the above observation, we can now prove Lemma~\ref{lem-first}.
We consider every $i \in [w]$ and compute the partition $\mathcal{R}_i$ of $S$.
Note that $\mathcal{R}_i$ can be computed in $O(n)$ time as we sorted the points in $S$ by their $\phi$-values.
For every $R \in \mathcal{R}_i$, the $\phi'$-values of the points in $R$ differ by at most $w$ and thus $\Delta_R \leq (\frac{3}{\varepsilon})^w = (\frac{3}{\varepsilon})^{\lceil \frac{3}{\varepsilon} \rceil}$.
Therefore, by our assumption, for every $R \in \mathcal{R}_i$, we can compute in $O_\varepsilon(|R|)$ time a $(1+\frac{\varepsilon}{2})$-approximation solution $E_R^* \subseteq E(R)$ for $\mathbf{Prob}(R)$.
The union $E_i^* = \bigcup_{R \in \mathcal{R}_i} E_R^*$ is a solution of $\mathbf{Prob}(S)$ and $\mathsf{cost}_S(E_i^*) = \sum_{R \in \mathcal{R}_i} \mathsf{cost}_R(E_R^*) \leq (1+\frac{\varepsilon}{2}) \sum_{R \in \mathcal{R}_i} \mathsf{opt}(R)$.
The total time for constructing $E_i^*$ is $O_\varepsilon(n)$, because $\sum_{R \in \mathcal{R}_i} |R| = n$.
We construct the solution $E_i^*$ for all $i \in [w]$ and finally output the best one among them.
Observation~\ref{obs-onegoodidx} guarantees the existence of $i \in [w]$ such that
\begin{equation*}
    \mathsf{cost}_S(E_i^*) \leq \left(1+\frac{\varepsilon}{2}\right) \sum_{R \in \mathcal{R}_i} \mathsf{opt}(E_R^*) \leq \left(1+\frac{\varepsilon}{2}\right) \left(1+\frac{\varepsilon}{3}\right) \cdot \mathsf{opt}(S) \leq (1+\varepsilon) \cdot \mathsf{opt}(S).
\end{equation*}
Therefore, our algorithm gives a $(1+\varepsilon)$-approximation solution for $\mathsf{opt}(S)$.
Since $w = O_\varepsilon(1)$, the total running time is still $O_\varepsilon(n)$.
This completes the proof of Lemma~\ref{lem-first}.


\subsection{Second reduction} \label{sec-second}

In this section, we further reduce a subproblem $\mathbf{Prob}(R)$ with bounded $\Delta_R$ to subproblems $\mathbf{Prob}(Q)$ where $Q$ has a small bounding box compared to the values $\phi(p)$ for $p \in Q$.
For $Q \subseteq S$, let $W_Q$ be the side-length of the smallest axis-parallel hypercube containing $Q$.
Our goal in this section is to prove the following lemma.

\begin{lemma} \label{lem-second}
Suppose for every subset $Q \subseteq S$ satisfying $W_Q \leq 2 \lceil \frac{4d}{\varepsilon} \rceil (\frac{3}{\varepsilon})^{\lceil \frac{3}{\varepsilon} \rceil} \cdot \min_{p \in Q} \phi(p)$, one can compute in $O_\varepsilon(|Q|)$ time a $(1+\frac{\varepsilon}{5})$-approximation solution for $\mathbf{Prob}(Q)$.
Then for every subset $R \subseteq S$ satisfying $\Delta_R \leq (\frac{3}{\varepsilon})^{\lceil \frac{3}{\varepsilon} \rceil}$, one can compute in $O_\varepsilon(|R|)$ time a $(1+\frac{\varepsilon}{2})$-approximation solution for $\mathbf{Prob}(R)$.
\end{lemma}

This reduction is done again by a shifting technique.
But this time, we apply grid shifting to the space $\mathbb{R}^d$.
As the points in $R$ have similar $\phi$-values, the edges in an optimal solution of $\mathbf{Prob}(R)$ also have similar lengths by Fact~\ref{fact-edgelen}.
This nice property allows us to use a shifted grid to decompose the problem by losing a factor of $1+O(\varepsilon)$ in cost.

Consider a subset $R \subseteq S$ satisfying $\Delta_R \leq (\frac{3}{\varepsilon})^{\lceil \frac{3}{\varepsilon} \rceil}$.
Let $E^* \subseteq E(R)$ be an optimal solution of $\mathbf{Prob}(R)$.
Set $r = \lceil \frac{4d}{\varepsilon} \rceil$, $\phi = \max_{p \in R} \phi(p)$, and $w = r \cdot 2\phi$.
We say an edge $e \in E^*$ is \textit{compatible} with a tuple $(k_1,\dots,k_d) \in [r]^d$ if the two endpoints of $e$ lie in the same cell of the $d$-dimensional grid $\varGamma_w(k_1 \cdot 2\phi,\dots,k_d \cdot 2\phi)$.
\begin{observation} \label{obs-manycomp}
Every $e \in E^*$ is compatible with at least $(r-1)^d$ tuples in $[r]^d$.
\end{observation}
\begin{proof}
Let $e \in E^*$.
As $\phi = \max_{p \in R} \phi(p)$, by Fact~\ref{fact-edgelen}, we have $|e| \leq 2 \phi$.
Thus, for every $i \in [d]$, there exists at most one integer $k_i^* \in \mathbb{N}$ such that the hyperplane $x_i = k_i^* \cdot 2\phi$ separates the two endpoints of $e$.
Note that $e$ is compatible with a tuple $(k_1,\dots,k_d) \in [r]^d$ iff for every $i \in [d]$ such that $k_i^*$ exists, $k_i$ and $k_i^*$ are not congruent modulo $r$.
For $i \in [d]$, define $K_i = [r]$ if $k_i^*$ does not exist and $K_i = \{k \in [r]: k \text{ is not congruent with } k_i^* \text{ modulo } r\}$ if $k_i^*$ exists.
Then the tuples which $e$ is compatible with are exactly those in $\prod_{i=1}^d K_i$.
We have $|K_i| \geq r-1$ for all $i \in [d]$.
Therefore, $|\prod_{i=1}^d K_i| \geq (r-1)^d$.
\end{proof}

For a tuple $\sigma = (k_1,\dots,k_d) \in [r]^d$, we denote by $\mathcal{Q}_\sigma$ the partition of $R$ induced by the grid $\varGamma_w(k_1 \cdot 2\phi,\dots,k_d \cdot 2\phi)$.

\begin{observation} \label{obs-onegoodtuple}
There exists a tuple $\sigma \in [r]^d$ such that $\sum_{Q \in \mathcal{Q}_\sigma}\mathsf{opt}(Q) \leq (1+\frac{\varepsilon}{4}) \cdot \mathsf{opt}(R)$.
\end{observation}
\begin{proof}
For $\sigma \in [r]^d$, denote by $N^*_\sigma = \{e \in E^*: e \text{ is not compatible with } \sigma\}$.
By Observation~\ref{obs-manycomp}, each $e \in E^*$ is contained in at most $r^d - (r-1)^d$ sets $N^*_\sigma$.
Thus, we have
\begin{equation*}
\sum_{\sigma \in [r]^d} \sum_{e \in N^*_\sigma} |e| \leq (r^d - (r-1)^d) \cdot \sum_{e \in E^*} |e| \leq (r^d - (r-1)^d) \cdot \mathsf{cost}_R(E^*).
\end{equation*}
So there exists some $\sigma \in [r]^d$ such that $\sum_{e \in N^*_\sigma} |e| \leq (r^d - (r-1)^d)/r^d \cdot \mathsf{cost}_R(E^*)$.
Note that $(r^d - (r-1)^d)/r^d \leq dr^{d-1}/r^d = d/r \leq \frac{\varepsilon}{4}$, which implies $\sum_{e \in N^*_\sigma} |e| \leq \frac{\varepsilon}{4} \cdot \mathsf{cost}_R(E^*)$.
We show that $\sigma$ satisfies the condition $\sum_{Q \in \mathcal{Q}_\sigma}\mathsf{opt}(Q) \leq (1+\frac{\varepsilon}{4}) \cdot \mathsf{opt}(R)$.
Clearly,
\begin{equation*}
    \sum_{Q \in \mathcal{Q}_\sigma}\mathsf{opt}(Q) \leq \sum_{Q \in \mathcal{Q}_\sigma} \mathsf{cost}_Q(E^* \cap E(Q)) \leq \mathsf{cost}_R(E^*) - \sum_{e \in N^*_\sigma} |e| + \sum_{p \in V(N^*_\sigma)} \phi(p).
\end{equation*}
Every edge $e = (p,q) \in E(S)$ satisfies $|e| \geq \max\{\phi(p),\phi(q)\} \geq \frac{1}{2} (\phi(p)+\phi(q))$ by the definition of $\phi$.
Therefore, $\sum_{e \in N^*_\sigma} |e| \geq \frac{1}{2} \sum_{p \in V(N^*_\sigma)} \phi(p)$.
It follows that
\begin{equation*}
    \mathsf{cost}_R(E^*) - \sum_{e \in N^*_\sigma} |e| + \sum_{p \in V(N^*_\sigma)} \phi(p) \leq \mathsf{cost}_R(E^*) + \sum_{e \in N^*_\sigma} |e|,
\end{equation*}
which implies $\sum_{Q \in \mathcal{Q}_\sigma}\mathsf{opt}(Q) \leq \mathsf{cost}_R(E^*) + \sum_{e \in N^*_\sigma} |e|$.
Since $\sum_{e \in N^*_\sigma} |e| \leq \frac{\varepsilon}{4} \cdot \mathsf{cost}_R(E^*)$, we finally have $\sum_{Q \in \mathcal{Q}_\sigma}\mathsf{opt}(Q) \leq (1+\frac{\varepsilon}{4}) \cdot \mathsf{cost}_R(E^*) = (1+\frac{\varepsilon}{4}) \cdot \mathsf{opt}(R)$.
\end{proof}

Using the above observation, we can now prove Lemma~\ref{lem-second}.
We consider every tuple $\sigma \in [r]^d$.
For each $\sigma \in [r]^d$, we first compute the partition $\mathcal{Q}_\sigma$ of $R$, which can be done in $O(|R|)$ time as the points in $R$ are sorted in every dimension.
For every $Q \in \mathcal{Q}_\sigma$,
\begin{equation*}
    W_Q \leq w = 2r\phi \leq 2 \left\lceil \frac{4d}{\varepsilon} \right\rceil \cdot \max_{p \in R} \phi(p) = 2 \left\lceil \frac{4d}{\varepsilon} \right\rceil \cdot \Delta_R \min_{p \in R} \phi(p),
\end{equation*}
and thus $W_Q \leq 2 \lceil \frac{4d}{\varepsilon} \rceil (\frac{3}{\varepsilon})^{\lceil \frac{3}{\varepsilon} \rceil} \cdot \min_{p \in Q} \phi(p)$ as $\Delta_R \leq (\frac{3}{\varepsilon})^{\lceil \frac{3}{\varepsilon} \rceil}$ and $\min_{p \in R} \phi(p) \leq \min_{p \in Q} \phi(p)$.
Therefore, by our assumption, for every $Q \in \mathcal{Q}_\sigma$, we can compute in $O_\varepsilon(|Q|)$ time a $(1+\frac{\varepsilon}{5})$-approximation solution $E_Q^* \subseteq E(Q)$ for $\mathbf{Prob}(Q)$.
The union $E_\sigma^* = \bigcup_{Q \in \mathcal{Q}_\sigma} E_Q^*$ is a solution of $\mathbf{Prob}(R)$ and $\mathsf{cost}_R(E_\sigma^*) = \sum_{Q \in \mathcal{Q}_\sigma} \mathsf{cost}_Q(E_Q^*) \leq (1+\frac{\varepsilon}{5}) \sum_{Q \in \mathcal{Q}_\sigma} \mathsf{opt}(Q)$.
The total time for constructing $E_\sigma^*$ is $O_\varepsilon(|R|)$, because $\sum_{Q \in \mathcal{Q}_\sigma} |Q| = |R|$.
We construct the solution $E_\sigma^*$ for all $\sigma \in [r]^d$ and finally output the best one among them.
Observation~\ref{obs-onegoodtuple} guarantees the existence of $\sigma \in [r]^d$ such that
\begin{equation*}
    \mathsf{cost}_R(E_\sigma^*) \leq \left(1+\frac{\varepsilon}{5}\right) \sum_{Q \in \mathcal{Q}_\sigma} \mathsf{opt}(Q) \leq \left(1+\frac{\varepsilon}{5}\right) \left(1+\frac{\varepsilon}{4}\right) \cdot \mathsf{opt}(R) \leq \left(1+\frac{\varepsilon}{2}\right) \cdot \mathsf{opt}(R).
\end{equation*}
Therefore, our algorithm gives a $(1+\frac{\varepsilon}{2})$-approximation solution for $\mathsf{opt}(R)$.
Since $r^d = O_\varepsilon(1)$, the total running time is still $O_\varepsilon(|R|)$.
This completes the proof of Lemma~\ref{lem-second}.

\subsection{Solving a well-structured subproblem} \label{sec-subprob}
By the reductions of Lemmas~\ref{lem-first} and~\ref{lem-second}, it now suffices to show that for every $Q \subseteq S$ satisfying $W_Q \leq 2 \lceil \frac{4d}{\varepsilon} \rceil (\frac{3}{\varepsilon})^{\lceil \frac{3}{\varepsilon} \rceil} \cdot \min_{p \in Q} \phi(p)$, one can compute in $O_\varepsilon(|Q|)$ time a $(1+\frac{\varepsilon}{5})$-approximation solution for $\mathbf{Prob}(Q)$.
In other words, our goal is to prove the following lemma.
\begin{lemma} \label{lem-restricted}
For every subset $Q \subseteq S$ satisfying $W_Q \leq 2 \lceil \frac{4d}{\varepsilon} \rceil (\frac{3}{\varepsilon})^{\lceil \frac{3}{\varepsilon} \rceil} \cdot \min_{p \in Q} \phi(p)$, one can compute in $O_\varepsilon(|Q|)$ time a $(1+\frac{\varepsilon}{5})$-approximation solution for $\mathbf{Prob}(Q)$.
\end{lemma}

The basic idea to solve such a subproblem $\mathbf{Prob}(Q)$ is the following.
Using the fact $\phi(p) = \Omega_\varepsilon(|W_Q|)$ for all $p \in Q$, we can partition the bounding box of $Q$ into $O_\varepsilon(1)$ small hypercubes such that the points in each hypercube have the same color and very similar $\phi$-values.
This allows us to view the points in a small hypercube as a single ``point'' and formulate an integer linear program with $O_\varepsilon(1)$ variables.
We then apply the FPT algorithm for integer linear programing~\cite{CyganFKLMPPS15} to solve the problem.
Below we discuss this in detail.

If all points in $Q$ have the same color, then $E(Q) = \emptyset$ and the subproblem $\mathbf{Prob}(Q)$ is trivial.
So assume $Q$ contains points of at least two colors.
Let $\Box$ be a hypercube containing $Q$ with side-length $w = 2 \lceil \frac{4d}{\varepsilon} \rceil (\frac{3}{\varepsilon})^{\lceil \frac{3}{\varepsilon} \rceil} \cdot \min_{p \in Q} \phi(p)$.
We set $r = 2d \lceil \frac{4d}{\varepsilon} \rceil (\frac{3}{\varepsilon})^{\lceil \frac{3}{\varepsilon} \rceil} \lceil \frac{44}{\varepsilon} \rceil$ and evenly partition $\Box$ into $r^d$ smaller hypercubes with side-length $\frac{w}{r}$.
Let $\mathcal{C}$ be the set of the smaller hypercubes which contain at least one point in $Q$.
We have the following observation.
\begin{observation} \label{obs-samecl}
    For every $C \in \mathcal{C}$, all points in $Q \cap C$ have the same color.
\end{observation}
\begin{proof}
Equivalently, we show no edge in $E(Q)$ has both endpoints in $C$.
Assume $e = (p,q)$ with $p,q \in Q \cap C$.
Then $|e| \leq \frac{dw}{r} < \phi(p)$, contradicting the fact that $|e| \geq \phi(p)$.
\end{proof}

By the above observation, for each $C \in \mathcal{C}$, we can define the \textit{color} of $C$, denoted by $\mathsf{cl}(C)$, as the color of the points in $Q \cap C$.
For a subset $E \subseteq E(Q)$, we define a function $f_E:\mathcal{C} \times \mathcal{C} \rightarrow \mathbb{N}$, where $f_E(C,C')$ is equal to the number of edges in $E$ whose one endpoint is in $C$ and the other endpoint is in $C'$.
Also, we define a function $g_E: \mathcal{C} \rightarrow \mathbb{N}$, where $g_E(C) = |Q \cap C| - |V(E) \cap C|$ is the number of points in $C$ not matched by $E$.

\begin{observation} \label{obs-sgnf}
    Let $f:\mathcal{C} \times \mathcal{C} \rightarrow \mathbb{N}$ and $g: \mathcal{C} \rightarrow \mathbb{N}$ be two functions.
    If there exists $E \subseteq E(Q)$ such that $f_E = f$ and $g_E = g$, then the following conditions hold.
    \smallskip    
    \begin{enumerate}
        \item For any $C,C' \in \mathcal{C}$, we have $f(C,C') = f(C',C) = 0$ if $\mathsf{cl}(C) = \mathsf{cl}(C')$, and $f(C,C') = f(C',C) \leq |Q \cap C| \cdot |Q \cap C'|$ if $\mathsf{cl}(C) \neq \mathsf{cl}(C')$.
        \smallskip
        \item For any $C \in \mathcal{C}$, we have $g(C) + \sum_{C' \in \mathcal{C}} f(C,C') \geq |Q \cap C|$.
    \end{enumerate}
    \smallskip
    Conversely, if $f$ and $g$ satisfy the conditions, then one can compute $E \subseteq E(Q)$ in $O_\varepsilon(|Q|)$ time such that $f_E(C,C') \leq f(C,C')$ for all $(C,C') \in \mathcal{C} \times \mathcal{C}$ and $g_E(C) \leq g(C)$ for all $C \in \mathcal{C}$.
\end{observation}
\begin{proof}
Suppose $f_E = f$ and $g_E = g$ for some $E \subseteq E(Q)$.
By the definition of $f_E$, it is clear that $f(C,C') = f(C',C)$ for any $C,C' \in \mathcal{C}$.
For any $C,C' \in \mathcal{C}$ with $\mathsf{cl}(C) = \mathsf{cl}(C')$, there cannot be any edge in $E$ with one endpoint is in $C$ and the other endpoint is in $C'$, and thus $f(C,C') = 0$.
For any $C,C' \in \mathcal{C}$ with $\mathsf{cl}(C) \neq \mathsf{cl}(C')$, there can be at most $|Q \cap C| \cdot |Q \cap C'|$ edges in $E$ with one endpoint is in $C$ and the other endpoint is in $C'$, and thus $f(C,C') \leq |Q \cap C| \cdot |Q \cap C'|$.
So condition~1 holds.
To see condition~2, observe that for any $C \in \mathcal{C}$, $|V(E) \cap C| \leq \sum_{C' \in \mathcal{C}} f_E(C,C')$.
Thus, by the definition of $g_E$, we directly have $g(C) + \sum_{C' \in \mathcal{C}} f(C,C') \geq g(C) + |V(E) \cap C| = |Q \cap C|$.

Now suppose $f$ and $g$ satisfy the two conditions.
For $(C,C') \in \mathcal{C} \times \mathcal{C}$ with $f(C,C') > |Q|$, we set $f(C,C') = |Q|$.
After this change, $f$ and $g$ still satisfy the conditions.
We construct the desired subset $E \subseteq E(Q)$ as follows.
Initially, set $E = \emptyset$.
For every $C,C' \in \mathcal{C}$ with $\mathsf{cl}(C) \neq \mathsf{cl}(C')$, we shall pick $f(C,C')$ edges in $E(Q)$ with one endpoint is in $C$ and the other endpoint is in $C'$, and add them to $E$; this is possible since $f(C,C') \leq |Q \cap C| \cdot |Q \cap C'|$ by condition~1.
The resulting $E$ guarantees $f_E = f$ (and thus the $f_E$-values are smaller than or equal to the original $f$-values).
To further guarantee $g_E(C) \leq g(C)$ for all $C \in \mathcal{C}$, we use the following rule to pick edges.
When we want to pick an edge $e$ with one endpoint in $C$ and the other endpoint in $C'$ (and add it to $E$), we always choose the endpoints of $e$ among the points in $Q \cap C$ and $Q \cap C'$ that are not the endpoints of the edges in the current $E$; if no such points exist, we then choose the endpoints of $e$ arbitrarily.
In this way, we can guarantee that at the end, for every $C \in \mathcal{C}$, either $|V(E) \cap C| = \sum_{C' \in \mathcal{C}} f(C,C')$ or $Q \cap C = V(E) \cap C$.
In the former case, $g_E(C) = |Q \cap C| - |V(E) \cap C| = |Q \cap C| - \sum_{C' \in \mathcal{C}} f(C,C') \leq g(C)$ by condition~2.
In the latter case, $g_E(C) = 0 \leq g(C)$.
This simple construction of $E$ can be done in $O_\varepsilon(|Q|)$ time since $f(C,C') \leq |Q|$ for all $(C,C') \in \mathcal{C} \times \mathcal{C}$.
\end{proof}

For convenience, we write $\phi(C) = \min_{p \in Q \cap C} \phi(p)$ for $C \in \mathcal{C}$.
For $C,C' \in \mathcal{C}$, let $\mathsf{dist}(C,C')$ denote the Euclidean distance between the centers of the hypercubes $C$ and $C'$.
Now for functions $f:\mathcal{C} \times \mathcal{C} \rightarrow \mathbb{N}$ and $g: \mathcal{C} \rightarrow \mathbb{N}$, we define
\begin{equation*}
    \mathsf{cost}(f,g) = \frac{1}{2} \sum_{(C,C') \in \mathcal{C} \times \mathcal{C}} (f(C,C') \cdot \mathsf{dist}(C,C')) + \sum_{C \in \mathcal{C}} (g(C) \cdot \phi(C)).
\end{equation*}

\begin{observation} \label{obs-sgnfcost}
    For every subset $E \subseteq E(Q)$, we have
    \begin{equation*}
        \left(1-\frac{\varepsilon}{21}\right) \mathsf{cost}(f_E,g_E) \leq \mathsf{cost}_Q(E) \leq \left(1+\frac{\varepsilon}{21}\right) \mathsf{cost}(f_E,g_E).
    \end{equation*}
    In particular, for any two subsets $E,E' \subseteq E(Q)$ satisfying $\mathsf{cost}(f_E,g_E) \leq c \cdot \mathsf{cost}(f_{E'},g_{E'})$, we have $\mathsf{cost}_Q(E) \leq (1+\frac{\varepsilon}{10}) c \cdot \mathsf{cost}_Q(E')$.
\end{observation}
\begin{proof}
We first show for any $e \in E(Q)$, $(1-\frac{\varepsilon}{21}) \cdot \mathsf{dist}(C,C') \leq |e| \leq (1+\frac{\varepsilon}{21}) \cdot \mathsf{dist}(C,C')$, where $C,C' \in \mathcal{C}$ are the hypercubes containing the two endpoints of $e$.
Suppose $e = (p,q)$ where $p \in C$ and $q \in C'$.
We have $|e| \geq \max\{\phi(p),\phi(q)\} \geq \frac{w}{r} \cdot d\lceil \frac{44}{\varepsilon} \rceil$.
Observe that the distance between any two points in $C$ (or $C'$) is at most $\frac{w}{r} \cdot d$.
Thus, the difference between $|e|$ and $\mathsf{dist}(C,C')$ is at most $(2/\lceil \frac{44}{\varepsilon} \rceil) \cdot |e| \leq \frac{\varepsilon}{22} \cdot |e|$, which implies $\mathsf{dist}(C,C') \geq \frac{21}{22} \cdot |e|$.
It follows that the difference between $|e|$ and $\mathsf{dist}(C,C')$ is at most $\frac{\varepsilon}{21} \cdot \mathsf{dist}(C,C')$.
So we have $(1-\frac{\varepsilon}{21}) \cdot \mathsf{dist}(C,C') \leq |e| \leq (1+\frac{\varepsilon}{21}) \cdot \mathsf{dist}(C,C')$.

Next, we show that for any $p \in Q$, $\phi(C) \leq \phi(p) \leq (1+\frac{\varepsilon}{21}) \cdot \phi(C)$, where $C \in \mathcal{C}$ is the hypercube containing $p$.
The inequality $\phi(C) \leq \phi(p)$ follows from the definition of $\phi(C)$.
To see $\phi(p) \leq (1+\frac{\varepsilon}{21}) \cdot \phi(C)$, suppose $\phi(C) = \phi(q)$ for $q \in Q \cap C$.
By Observation~\ref{obs-samecl}, $\mathsf{cl}(p) = \mathsf{cl}(q)$.
This implies $|\phi(p) - \phi(q)| \leq d \cdot \frac{w}{r}$, since the side-length of $C$ is $\frac{w}{r}$.
Note that $d \cdot \frac{w}{r} \leq \frac{\varepsilon}{44} \phi(q) \leq \frac{\varepsilon}{21} \phi(q)$.
Therefore, $\phi(p) \leq (1+\frac{\varepsilon}{21}) \cdot \phi(q)$.

Now we prove the observation.
For each $p \in Q$, we denote by $C_p \in \mathcal{C}$ the hypercube containing $p$.
Let $E = \{e_1,\dots,e_m\} \subseteq E(Q)$ and suppose $e_i = (p_i,q_i)$ for $i \in [m]$.
Then we have $\mathsf{cost}_Q(E) = \sum_{i=1}^m |e_i| + \sum_{p \in Q \backslash V(E)} \phi(p)$.
As shown above, $(1-\frac{\varepsilon}{21}) \cdot \mathsf{dist}(C_{p_i},C_{q_i}) \leq |e_i| \leq (1+\frac{\varepsilon}{21}) \cdot \mathsf{dist}(C_{p_i},C_{q_i})$ for all $i \in [m]$ and $\phi(C_p) \leq \phi(p) \leq (1+\frac{\varepsilon}{21}) \cdot \phi(C_p)$ for all $p \in Q \backslash V(E)$.
Therefore, if we set $\alpha = \sum_{i=1}^m \mathsf{dist}(C_{p_i},C_{q_i}) + \sum_{p \in Q \backslash V(E)} \phi(C_p)$, then we have $(1-\frac{\varepsilon}{21}) \cdot \alpha \leq \mathsf{cost}_Q(E) \leq (1+\frac{\varepsilon}{21}) \cdot \alpha$.
It suffices to show $\alpha = \mathsf{cost}(f_E,g_E)$.
By the definitions of the functions $f_E$ and $g_E$, we have
\begin{align*}
    \alpha &= \frac{1}{2}\sum_{i=1}^m (\mathsf{dist}(C_{p_i},C_{q_i})+\mathsf{dist}(C_{q_i},C_{p_i})) + \sum_{p \in Q \backslash V(E)} \phi(C_p) \\
    & = \frac{1}{2} \sum_{(C,C') \in \mathcal{C} \times \mathcal{C}} (f_E(C,C') \cdot \mathsf{dist}(C,C')) + \sum_{C \in \mathcal{C}} (g_E(C) \cdot \phi(C)).
\end{align*}
Thus, $\alpha = \mathsf{cost}(f_E,g_E)$.
To see the second statement in the observation, let $E,E' \subseteq E(Q)$ satisfying $\mathsf{cost}(f_E,g_E) \leq c \cdot \mathsf{cost}(f_{E'},g_{E'})$.
Write $\alpha = \mathsf{cost}(f_E,g_E)$ and $\alpha' = \mathsf{cost}(f_{E'},g_{E'})$ for convenience.
It follows that 
\begin{equation*}
    \mathsf{cost}_Q(E) \leq \left(1+\frac{\varepsilon}{21}\right) \cdot \alpha \leq \left(1+\frac{\varepsilon}{21}\right)c \cdot \alpha' \leq \frac{(1+\frac{\varepsilon}{21})c}{1-\frac{\varepsilon}{21}} \cdot \mathsf{cost}_Q(E').
\end{equation*}
As $(1+\frac{\varepsilon}{21})/(1-\frac{\varepsilon}{21}) \leq 1+\frac{\varepsilon}{10}$, we have $\mathsf{cost}_Q(E) \leq (1+\frac{\varepsilon}{10}) c \cdot \mathsf{cost}_Q(E')$.
\end{proof}

Thanks to Observation~\ref{obs-sgnf} and~\ref{obs-sgnfcost}, we can use the following idea to compute a $(1+\frac{\varepsilon}{5})$-approximation solution for $\mathbf{Prob}(Q)$.
We say a pair $(f,g)$ of functions $f: \mathcal{C} \times \mathcal{C} \rightarrow \mathbb{N}$ and $g: \mathcal{C} \rightarrow \mathbb{N}$ is \textit{valid} if it satisfies the two conditions in Observation~\ref{obs-sgnf}.
First, we compute a valid pair $(f,g)$ satisfying $\mathsf{cost}(f,g) \leq (1+\frac{\varepsilon}{11}) \cdot \mathsf{cost}(f',g')$ for any valid pair $(f',g')$.
We shall show later how to efficiently compute such a pair $(f,g)$ by solving an integer linear program.
Then we use Observation~\ref{obs-sgnf} to compute in $O_\varepsilon(|Q|)$ time a subset $E \subseteq E(Q)$ satisfying $f_E(C,C') \leq f(C,C')$ for all $(C,C') \in \mathcal{C} \times \mathcal{C}$ and $g_E(C) \leq g(C)$ for all $C \in \mathcal{C}$.
We claim that $E$ is a $(1+\frac{\varepsilon}{5})$-approximation solution of $\mathbf{Prob}(Q)$.
To see this, suppose $E^* \subseteq E(Q)$ is an optimal solution of $\mathbf{Prob}(Q)$.
The pair $(f_{E^*},g_{E^*})$ is valid by Observation~\ref{obs-sgnf}.
So we have $\mathsf{cost}(f_E,g_E) \leq \mathsf{cost}(f,g) \leq (1+\frac{\varepsilon}{11}) \cdot \mathsf{cost}(f_{E^*},g_{E^*})$.
Then by Observation~\ref{obs-sgnfcost},
\begin{equation*}
    \mathsf{cost}_Q(E) \leq \left(1+\frac{\varepsilon}{10}\right)\left(1+\frac{\varepsilon}{11}\right) \cdot \mathsf{cost}_Q(E^*) \leq \left(1+\frac{\varepsilon}{5}\right) \cdot \mathsf{cost}_Q(E^*).
\end{equation*}

Now we show how to compute a valid pair $(f,g)$ satisfying $\mathsf{cost}(f,g) \leq (1+\frac{\varepsilon}{11}) \cdot \mathsf{cost}(f',g')$ for any valid pair $(f',g')$.
This is done by formulating an integer linear program with $O_\varepsilon(1)$ variables and applying the FPT algorithm for integer linear programming~\cite{CyganFKLMPPS15}.
We view the values $f(C,C')$ for $(C,C') \in \mathcal{C} \times \mathcal{C}$ and $g(C)$ for $C \in \mathcal{C}$ as integer variables.
The objective function to be minimized is $\mathsf{cost}(f,g)$, which is a linear function of the variables.
To check if $(f,g)$ is valid, it is equivalent to check if $(f,g)$ satisfies the two conditions in Observation~\ref{obs-sgnf}, which can be described as linear constraints on the variables.
Therefore, finding a valid pair $(f,g)$ with minimum $\mathsf{cost}(f,g)$ is equivalent to assigning (non-negative) integer values to the variables to minimize the objective function under the linear constraints.
Note that this is not exactly an integer linear program, because the coefficients of the objective function are real numbers (while the coefficients of the linear constraints are all integers).
However, as we only need a valid pair $(f,g)$ with \textit{approximately} minimum $\mathsf{cost}(f,g)$, we can round these real coefficients to integers without changing the program too much.
Observe that for any distinct $C,C' \in \mathcal{C}$, $\mathsf{dist}(C,C') \geq \frac{w}{r}$.
Furthermore, $\phi(p) \geq \frac{w}{r}$ for all $p \in Q$ and thus $\phi(C) \geq \frac{w}{r}$ for all $C \in \mathcal{C}$.
We replace every coefficient $\eta$ in the objective function with a new coefficient $\lfloor \lceil \frac{11}{\varepsilon} \rceil \cdot \frac{\eta}{w/r} \rfloor$ and obtain a new objective function
\begin{equation*}
    \mathsf{cost}'(f,g) = \sum_{(C,C') \in \mathcal{C} \times \mathcal{C}} \left(f(C,C') \cdot \left\lfloor \left\lceil \frac{11}{\varepsilon} \right\rceil \cdot \frac{\mathsf{dist}(C,C')}{w/r} \right\rfloor \right) + \sum_{C \in \mathcal{C}} \left(g(C) \cdot \left\lfloor \left\lceil \frac{11}{\varepsilon} \right\rceil \cdot \frac{\phi(C)}{w/r} \right\rfloor \right).
\end{equation*}

\begin{observation}
Let $(f,g)$ be a valid pair with minimum $\mathsf{cost}'(f,g)$.
Then $\mathsf{cost}(f,g) \leq (1+\frac{\varepsilon}{11}) \cdot \mathsf{cost}(f',g')$ for any valid pair $(f',g')$.
\end{observation}
\begin{proof}
We first claim that $\lceil \frac{11}{\varepsilon} \rceil \cdot \frac{\eta}{w/r} \leq (1+\frac{\varepsilon}{11}) \cdot \lfloor \lceil \frac{11}{\varepsilon} \rceil \cdot \frac{\eta}{w/r} \rfloor$ for any coefficient $\eta$ of the $\mathsf{cost}$ function.
As observed before, $\eta \geq \frac{w}{r}$ and thus $\lfloor \lceil \frac{11}{\varepsilon} \rceil \cdot \frac{\eta}{w/r} \rfloor \geq \lceil \frac{11}{\varepsilon} \rceil$.
So we have
\begin{equation*}
    \left\lceil \frac{11}{\varepsilon} \right\rceil \cdot \frac{\eta}{w/r} \leq \left\lfloor \left\lceil \frac{11}{\varepsilon} \right\rceil \cdot \frac{\eta}{w/r} \right\rfloor + 1 \leq \left(1+\frac{\varepsilon}{11}\right) \cdot \left\lfloor \left\lceil \frac{11}{\varepsilon} \right\rceil \cdot \frac{\eta}{w/r} \right\rfloor.
\end{equation*}
This implies that $\lceil \frac{11}{\varepsilon} \rceil \cdot \frac{\mathsf{cost}(f,g)}{w/r} \leq (1+\frac{\varepsilon}{11}) \cdot \mathsf{cost}'(f,g)$.
By the choice of $(f,g)$, it then follows that $\lceil \frac{11}{\varepsilon} \rceil \cdot \frac{\mathsf{cost}(f,g)}{w/r} \leq (1+\frac{\varepsilon}{11}) \cdot \mathsf{cost}'(f',g')$ for any valid pair $(f',g')$.
By the definition of $\mathsf{cost}'$, we directly have $\mathsf{cost}'(f',g') \leq \lceil \frac{11}{\varepsilon} \rceil \cdot \frac{\mathsf{cost}(f',g')}{w/r}$.
Therefore, for any valid pair $(f',g')$, $\lceil \frac{11}{\varepsilon} \rceil \cdot \frac{\mathsf{cost}(f,g)}{w/r} \leq (1+\frac{\varepsilon}{11}) \cdot \lceil \frac{11}{\varepsilon} \rceil \cdot \frac{\mathsf{cost}(f',g')}{w/r}$ and thus $\mathsf{cost}(f,g) \leq (1+\frac{\varepsilon}{11}) \cdot \mathsf{cost}(f',g')$.
\end{proof}

With the above observation, it suffices to minimize the new objective function $\mathsf{cost}'(f,g)$, which is a linear function of the variables with integer coefficients.
Therefore, our task becomes solving an integer linear program.
The numbers of variables and constraints are both $O_\varepsilon(1)$.
The coefficients in the linear constraints are bounded by $|Q|^{O(1)}$.
To bound the coefficients in the objective function, observe that $\mathsf{dist}(C,C') \leq d w$ for any $C,C' \in \mathcal{C}$ (because the side-length of $\Box$ is $w$) and thus $\frac{\mathsf{dist}(C,C')}{w/r} = O_\varepsilon(1)$.
Also, since $Q$ contains points of at least two colors, every point $p \in Q$ has a foreign neighbor in $Q$ (which has distance at most $dw$ to $p$) and thus $\phi(p) \leq d w$.
It follows that $\phi(C) \leq d w$ for all $C \in \mathcal{C}$ and thus $\frac{\phi(C)}{w/r} = O_\varepsilon(1)$.
So the coefficients in the objective function are bounded by $O_\varepsilon(1)$.
The entire integer linear program can be encoded in $O_\varepsilon(\log |Q|)$ bits.
It is well-known~\cite{CyganFKLMPPS15} that an integer linear program encoded in $L$ bits with $s$ variables can be solved in $s^{O(s)} \cdot L^{O(1)}$ time.
Therefore, our program can be solved in $O_\varepsilon(\log^{O(1)} |Q|)$ time.
Including the time for constructing the program and finding the set $E$ using Observation~\ref{obs-sgnf}, the total time cost is $O_\varepsilon(|Q|)$, which proves Lemma~\ref{lem-restricted}.

\subsection{Putting everything together} \label{sec-alltogether}

Combining Lemmas~\ref{lem-first}, \ref{lem-second}, and~\ref{lem-restricted}, we can compute a $(1+\varepsilon)$-approximation solution for $\mathbf{Prob}(S)$ in $O_\varepsilon(n)$ time.
Including the time for pre-sorting, we see that geometric many-to-many matching in any fixed dimension admits a $(1+\varepsilon)$-approximation algorithm with running time $O_\varepsilon(n \log n)$, assuming the nearest foreign neighbors of the points are given.

Finally, we provide the last missing piece of our result: how to solve the problem without knowing the nearest foreign neighbors.
While it is well-known that the all-nearest-neighbor problem can be solved in $O(n \log n)$ time in any fixed dimension~\cite{vaidya1989n}, computing all nearest foreign neighbors is much more challenging: it admits an $O(n \log n)$-time algorithm only for $d = 2$~\cite{aggarwal1992optimal} and has a conjectured $\Omega(n^{4/3})$ lower bound for $d \geq 3$~\cite{erickson1995relative}.
However, as we only want to compute an approximation solution, we can actually use $(1+\varepsilon)$-approximate nearest foreign neighbors (instead of the exact nearest foreign neighbors) which can be computed in $O_\varepsilon(n \log n)$ time by Lemma~\ref{lem-anfn}.
For a point $p \in S$, let $\mathsf{ann}(p) \in S$ be a $(1+\varepsilon)$-approximate nearest foreign neighbor of $p$ in $S$.
We use $\phi'(p) = \frac{\mathsf{dist}(p,\mathsf{ann}(p))}{1+\varepsilon}$ as the penalty of a point $p \in S$, where $\mathsf{dist}$ denotes the Euclidean distance.
Note that $\phi'(p) \leq \phi(p) \leq (1+\varepsilon) \cdot \phi'(p)$.
It is easy to see the following analogy of Lemma~\ref{lem-penalized}.

\begin{lemma}
Given a subset $E \subseteq E(S)$, one can compute in $O(n+|E|)$ time another subset $E' \subseteq E(S)$ such that $V(E') = S$ and $\sum_{e \in E'} |e| \leq (1+\varepsilon) \cdot (\sum_{e \in E} |e| + \sum_{p \in S \backslash V(E)} \phi'(p))$.
\end{lemma}
\begin{proof}
Like what we did in the proof of Lemma~\ref{lem-penalized}, we simply set 
\begin{equation*}
    E' = E \cup \{(p,\mathsf{ann}(p)): p \in S \backslash V(E)\}.    
\end{equation*}
Clearly, $V(E') = S$.
For an edge $e = (p,\mathsf{ann}(p))$, we have $|e| = (1+\varepsilon) \cdot \phi(p)$.
Thus, $\sum_{e \in E'} |e| \leq (1+\varepsilon) \cdot (\sum_{e \in E} |e| + \sum_{p \in S \backslash V(E)} \phi'(p))$.
\end{proof}

Now it suffices to consider the penalized formulation with the new penalty function $\phi'$.
The same algorithm still works.
Indeed, in our algorithm, we only use two properties of the old penalty function $\phi$.
First, we need $|e| \geq \max\{\phi(p),\phi(q)\}$ for every edge $e = (p,q) \in E(S)$.
This also holds for $\phi'$.
Second, in the proof of Observation~\ref{obs-sgnfcost}, we need the inequality that for any $p \in Q$, $\phi(C) \leq \phi(p) \leq (1+\frac{\varepsilon}{21}) \cdot \phi(C)$, where $C \in \mathcal{C}$ is the hypercube containing $p$ and $\phi(C) = \min_{p \in Q \cap C} \phi(p)$.
When replacing $\phi$ with $\phi'$ and setting $\phi'(C) = \min_{p \in Q \cap C} \phi'(p)$, we still have $\phi'(C) \leq \phi'(p) \leq (1+O(\varepsilon)) \cdot \phi'(C)$, because $\phi'(p) \leq \phi(p) \leq (1+\varepsilon) \cdot \phi'(p)$.
Therefore, applying the same algorithm with the new penalty function $\phi'$, we can obtain a $(1+O(\varepsilon))$-approximation solution, which is sufficient for the purpose of an approximation scheme, since we can always choose $\varepsilon$ to be smaller than the required approximation ratio by a constant factor.
This completes the proof of Theorem~\ref{thm-main} for the Euclidean case.

\subparagraph{Generalization to $L_p$-norms.}
Our algorithm directly applies to geometric many-to-many matching under the $L_p$-norm for any $p \geq 1$.
The only thing we need to adjust is the parameter $r$ in Section~\ref{sec-subprob}: for different norms, we need to partition the hypercube $\Box$ into different numbers of smaller hypercubes to make Observations~\ref{obs-samecl} and~\ref{obs-sgnfcost} hold.

%% file: conclusion.tex
In this paper, we studied the geometric many-to-many matching problem.
We give a $(1+\varepsilon)$-approximation algorithm with running time $O_\varepsilon(n \log n)$ for geometric many-to-many matching in any fixed dimension under the $L_p$-norm for any $p \geq 1$.
Our result significantly improves and generalizes the previous work on the problem.

We pose several open questions for future study.
First, the running time of our algorithm has an exponential dependency on $\frac{1}{\varepsilon}$, which comes from both the reduction in Section~\ref{sec-first} and the FPT algorithm for integer linear programming used in Section~\ref{sec-subprob}.
It is interesting to see whether one can improve the bound to $(\frac{1}{\varepsilon})^{O(1)} \cdot n \log n$.
Second, we only explored geometric many-to-many matching in the context of approximation algorithms.
Designing fast exact algorithms for the problem seems much more difficult and is an important open problem to be studied.
Finally, as mentioned in the introduction, several variants of geometric many-to-many have been studied in the literature~\cite{Rajabi-AlniB16, abs-1904-05184, abs-1904-03015}.
In these variants, each input point $p \in S$ is associated with a demand $\alpha(p)$ and/or a capacity $\beta(p)$, and a solution $E \subseteq E(S)$ should satisfy that
\begin{equation*}
\alpha(p) \leq |\{e \in E: e \text{ is incident to } p\}| \leq \beta(p)    
\end{equation*}
for all $p \in S$.
The previous study on these variants is restricted in $\mathbb{R}^1$.
Thus, it is natural to ask whether these variants of geometric many-to-many matching also admit approximation schemes with near-linear running time in any fixed dimension and whether the techniques in this paper can be adapted (together with other ideas) to design such algorithms.